\documentclass[8pt,onecolumn]{IEEEtran}
\usepackage{amsmath}
\usepackage{amssymb}
\usepackage{graphics}
\usepackage{epstopdf}
\usepackage{graphicx}
\usepackage{epsfig,amssymb}
\usepackage{amsmath}
\usepackage{mathrsfs}
\usepackage{color}
\usepackage{array}
\usepackage{amsfonts,booktabs}
\usepackage{lipsum,amsmath,multicol}
\usepackage{graphicx}
\usepackage{subfigure}
\usepackage{times}
\usepackage{pdfsync}
\usepackage{pgfplots}
\usepackage{pgfplotstable}
\usepackage[subnum]{cases}
\usepackage{multirow}
\usepackage{algorithm}
\usepackage{algpseudocode}
\usepackage{float}
\usepackage{setspace}
\usepackage{bm}

\usepackage{amsthm}
\usepackage{epstopdf}

\usepackage{stmaryrd}
\usepackage{hyperref}
\hypersetup{hidelinks=true}
\usepackage{textcomp}
\def\BibTeX{{\rm B\kern-.05em{\sc i\kern-.025em b}\kern-.08em
		T\kern-.1667em\lower.7ex\hbox{E}\kern-.125emX}}
\usepackage{cite}
\usepackage{amsmath,amssymb,amsfonts}
\usepackage{graphicx}
\usepackage{textcomp}

\usepackage{graphics}
\usepackage{epstopdf}
\usepackage{epsfig}
\usepackage{mathrsfs}
\usepackage{color}
\usepackage{array}
\usepackage{booktabs}
\usepackage{lipsum,multicol}
\usepackage{subfigure}
\usepackage{times}

\usepackage{algorithm}
\usepackage{pdfsync}
\usepackage{pgfplots}
\usepackage{pgfplotstable}
\usepackage[subnum]{cases}
\usepackage{multirow}
\usepackage{algpseudocode}
\usepackage{float}
\usepackage{setspace}
\usepackage{bm}
\usepackage{multicol}
\usepackage{epstopdf}
\usepackage[center]{caption}
\newtheorem{theorem}{Theorem}
\newtheorem{lemma}{Lemma}

\newcommand{\beq}{\begin{equation}}
	\newcommand{\eeq}{\end{equation}}

\newcommand{\abs}[1]{\left|#1\right|} 

\def\BibTeX{{\rm B\kern-.05em{\sc i\kern-.025em b}\kern-.08em
		T\kern-.1667em\lower.7ex\hbox{E}\kern-.125emX}}
\begin{document}
	\title{Optimal Real-Time Fusion of Time-Series Data Under R\'enyi Differential Privacy}
	\author{Chuanghong Weng, Ehsan Nekouei \thanks{
			C. Weng and E. Nekouei are with the Department of Electrical Engineering, City University of Hong Kong (e-mail: cweng7-c@cityu.edu.hk; enekouei@cityu.edu.hk). }		
		\thanks{The work was partially supported by the Research Grants Council of Hong Kong under Project CityU 21208921, a grant from Chow Sang Sang Group Research Fund sponsored by Chow Sang Sang Holdings International Limited.}}
	
	\maketitle
	
	\begin{abstract}
		In this paper, we investigate the optimal real-time fusion of data collected by multiple sensors. In our set-up, the sensor measurements are considered to be private and are jointly correlated with an underlying process. A fusion center combines the private sensor measurements and releases its output to an honest-but-curious party, which is responsible for estimating the state of the underlying process based on the fusion center's output. The privacy leakage incurred by the fusion policy is quantified using Rényi differential privacy. We formulate the privacy-aware fusion design as a constrained finite-horizon optimization problem, in which the fusion policy and the state estimation are jointly optimized to minimize the state estimation error subject to a total privacy budget constraint. We derive the constrained optimality conditions for the proposed optimization problem and use them to characterize the structural properties of the optimal fusion policy. Unlike classical differential privacy mechanisms, the optimal fusion policy is shown to adaptively allocates the privacy budget and regulates the adversary’s belief in a closed-loop manner. To reduce the computational burden of solving the resulting constrained optimality equations, we parameterize the fusion policy using a structured Gaussian distribution and show that the parameterized fusion policy satisfies the privacy constraint. We further develop a numerical algorithm to jointly optimize the fusion policy and state estimator. Finally, we demonstrate the effectiveness of the proposed fusion framework through a traffic density estimation case study.
	\end{abstract}
	
	\begin{IEEEkeywords}
		Sensor fusion, R\'enyi privacy, differential privacy, optimization
	\end{IEEEkeywords}
	
	\section{Introduction}\label{sec:introduction}
	\subsection{Motivation}	
	Information fusion from multiple sensors is a fundamental technique for environmental perception, monitoring, and control. By integrating heterogeneous measurements, a fused estimate can achieve higher accuracy, robustness, and reliability than any individual sensor alone. However, multi-source data collection and fusion often involve sensitive information, which raises significant privacy concerns. For example, traffic density estimation commonly relies on vehicles' position and velocity data that can be automatically collected by smart devices, e.g., GPS-enabled smartphones. Although such measurements are critical for real-time traffic management and congestion prediction, they may inadvertently disclose sensitive personal attributes, such as a driver’s home and workplace locations, thereby exposing private behavioral patterns. To mitigate these risks, data collected from multiple sources are typically aggregated and processed through privacy-preserving mechanisms before being utilized for downstream estimation and control tasks.
	
	
	Due to its strong mathematical guarantees and ease of implementation, differential privacy (DP) has been widely adopted for privacy-preserving data disclosure. Under DP, a privacy mechanism releases perturbed data at the cost of privacy leakage and halts data sharing once a prescribed privacy budget is exhausted. In online sensor fusion, where measurements arrive sequentially, the fusion policy should adaptively regulate its behavior and appropriately allocate the privacy budget over time to ensure sustained privacy protection while maximizing estimation performance. Otherwise, the utility of the fusion result may be significantly degraded if the privacy budget is allocated improperly or if the fusion policy is poorly designed. Despite its practical importance, optimal sensor fusion with adaptive privacy budget allocation for general nonlinear systems has not yet been thoroughly investigated in the existing literature.
	\subsection{Related Work}
	Privacy-preserving methods for dynamical systems have been primarily investigated from two perspectives: information-theoretic privacy and differential privacy. Information-theoretic approaches typically quantify privacy leakage using metrics such as conditional entropy and mutual information, and formulate privacy-aware estimation and control as constrained optimization problems that balance utility and privacy \cite{nekouei2019information}. In \cite{mochaourab2018private}, the authors designed a privacy filter to protect private states in hidden Markov models by optimizing the trade-off between monitoring accuracy and privacy leakage measured by the conditional entropy of the private states conditioned on the filter output. The works in \cite{tanaka2017directed, nekouei2019information} investigated state privacy protection in optimal linear control using directed information as the privacy metric, and showed that the optimal privacy mechanism consists of a Kalman filter followed by an additive noise module.
	
	The authors of \cite{erdemir2020privacy} studied the optimal privacy-aware data-sharing policy for Markov systems and proposed an actor–critic algorithm for policy optimization. In \cite{li2018information}, privacy leakage in smart metering systems was mitigated by embedding a rechargeable battery, and the optimal battery charging policy was derived using dynamic programming. The state trajectory obfuscation problem in partially observable Markov decision processes (POMDPs) was investigated in \cite{molloy2023smoother}, where it was shown that the optimal control law remains equivalent to the standard POMDP solution, and a piecewise-linear concave approximation approach was proposed to obtain bounded-error solutions. The optimal stochastic sampling approach for protecting private inputs in dynamic systems was studied in \cite{weng2025optimal}, which demonstrated that the privacy-aware sampler regulates the adversary’s belief in a closed-loop manner. However, as noted in \cite{weng2025optimal, weng2025optimalState}, the evaluation of information-theoretic privacy metrics is often computationally demanding, which renders exact optimization intractable in many practical settings.
	
	Differential privacy (DP) \cite{dwork2006differential} has also been introduced into dynamical systems to mitigate privacy leakage arising from the release of sensor measurements \cite{le2013differentially}. In \cite{le2013differentially}, the authors proposed a differentially private Kalman filter by injecting additive Gaussian noise to protect private inputs, which was later extended to differentially private sensor fusion for linear systems in \cite{yan2022guaranteeing}. Differentially private linear–quadratic control for multi-agent systems was investigated in \cite{yazdani2022differentially}, where Gaussian noise was added to measurements to protect agents’ states. In \cite{nozari2017differentially}, a Laplacian consensus algorithm was developed to preserve the privacy of agents’ initial states. The work in \cite{kawano2020design} established connections between differential privacy and input observability for linear systems and proposed a privacy-preserving controller design to protect tracking errors. In \cite{degue2022differentially}, the authors developed a differentially private Kalman filter through measurement aggregation and Gaussian noise injection, and showed that the optimal data aggregator minimizing the estimation error can be obtained by solving a semidefinite program. Privacy-preserving traffic density estimation under a known traffic model was studied in \cite{vishnoi2022differential}, where a moving-horizon estimation framework was proposed and validated using traffic micro-simulation software. A comprehensive and rigorous survey of differential privacy for fusion and estimation in dynamical systems was recently provided in \cite{yan2024privacy}.
	
	However, most existing DP mechanisms for dynamical systems rely on assumptions of linearity and Gaussian disturbances, while optimal differentially private estimation and control for general nonlinear systems remain largely open. Moreover, in many existing works, the privacy mechanism and its parameters, such as the maximum allowable per-step privacy leakage, are assumed to be time-invariant or fixed a priori when interacting with real data \cite{nozari2017differentially, kawano2020design, yazdani2022differentially}. Recent studies \cite{rogers2016privacy, feldman2021individual, whitehouse2023fully} have shown that adaptively selecting mechanisms and budgeting can significantly improve data utility under a fixed total privacy budget in the framework of R\'enyi differential privacy. It was also noted in \cite{liang2026adaptive} that a budgeting scheme for sequential queries, adapted to the outputs of previous queries, can improve the trade-off between data utility and privacy. Despite these advances, adaptive DP mechanisms and privacy budget allocation strategies that incorporate the dynamics of general systems have received limited attention in the literature.
	\begin{figure}[h]
		\centering
		\includegraphics[width=0.6\textwidth]{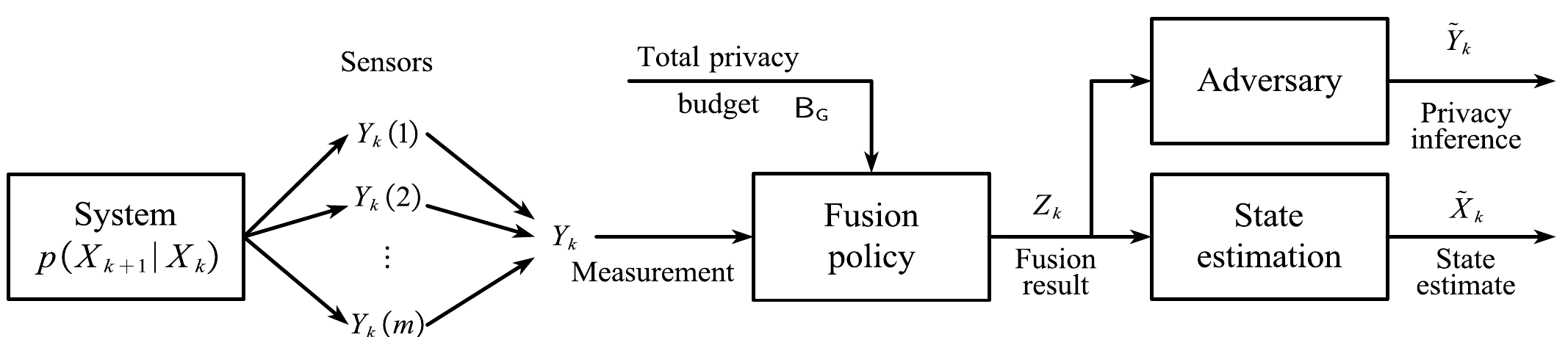}
		\caption{Privacy-aware sensor fusion with adaptive budget allocation}\label{Fig.SystemSetup}
	\end{figure}
	\subsection{Contributions}
	In this paper, we propose an optimal privacy-aware sensor fusion framework for general nonlinear systems, as illustrated in Fig.~\ref{Fig.SystemSetup}. The system state $X_k$ is observed by $m$ spatially distributed sensors, and the resulting measurements are collected in the observation vector $Y_k = [Y_k(i)]_{i=1}^m$, where $Y_k(i)$ denotes the measurement from the $i$th sensor. These measurements are processed by a fusion center to generate the fused output $Z_k$, which is then used by a remote estimator to reconstruct the system state. However, the fusion center’s output may also be exploited by an adversary to infer private information, thereby causing privacy leakage. The maximum allowable privacy leakage over the operational horizon $K$ is referred to as the privacy budget, denoted by $\mathsf{B_G} \in \mathbb{R}^{+}$. To mitigate privacy leakage arising from direct data disclosure, we design a privacy-aware fusion policy that releases $Z_k$ online while ensuring that the total privacy leakage does not exceed the budget $\mathsf{B_G}$.
	
	The main contribution of this work is the development of an optimal differentially private sensor fusion framework with adaptive privacy budget allocation. Specifically, the fusion policy and the state estimator are jointly designed to minimize the state estimation error under a hard constraint on the total allowable privacy leakage $\mathsf{B_G}$ over a finite horizon, where the privacy leakage is quantified using R\'enyi differential privacy (RDP). We derive the constrained optimality equations for the optimal fusion design and show that the proposed fusion policy automatically allocates the privacy budget $\mathsf{B_G}$ over time in a closed-loop manner. We further develop a practical numerical algorithm to implement the proposed fusion design. The fusion policy is parameterized as a structured Gaussian mechanism that first filters measurements from individual sensors and then aggregates the filtered outputs using an adaptive fusion vector. The RDP condition of the proposed fusion policy is explicitly derived and enforced numerically through value clipping. The filtering function and the state estimator are jointly optimized by minimizing the mean estimation error over sampled trajectories, while the fusion vector is optimized using a proximal policy optimization algorithm to enable adaptive privacy budget allocation.
	
	Compared with most existing differentially private mechanisms for linear Gaussian systems, \emph{e.g.}, \cite{nozari2017differentially, kawano2020design, yazdani2022differentially}, the proposed optimal sensor fusion design not only minimizes the state estimation error under a global differential privacy constraint for general nonlinear systems, but also adaptively determines its fusion output and allocates the privacy budget based on real-time measurements and past decisions. Furthermore, the effectiveness of the proposed privacy-aware design is demonstrated using the U.S. Highway 101 traffic dataset. Experimental results on traffic flow density estimation validate that the proposed method achieves a desirable trade-off between privacy protection and estimation accuracy.
	\section{Outline}
	The remainder of this paper is organized as follows. Section~\ref{Sec.ProbFormulation} introduces the system model and formulates the optimization problem for privacy-aware fusion design. Section~\ref{Sec.StructuralProperty} investigates the structural properties of the proposed fusion design. Section~\ref{Sec.NumericalAlg} develops a numerical algorithm to optimize the privacy-aware fusion design. Simulation results demonstrating the effectiveness of the proposed fusion policy are presented in Section~\ref{Sec.Simulation}, followed by conclusions in Section~\ref{Sec.Conclusion}.
	
	\section{Notation}
	We use uppercase letters, e.g., $X$, $Y$, to denote random variables, and lowercase letters, e.g., $x$, $y$, to denote their realizations. The shorthand $X^K$ represents the sequence $[X_0, X_1, \dots, X_K]$. The symbols $p(X)$ and $p(X|Y)$ denote the probability density function and the conditional probability density function, respectively. The expectation of an function $f\left(X\right)$ is denoted with $\mathsf{E}\left[f(X)\right]=\int{p\left(x\right)f(x)dx}$.
	
	\section{Problem Formulation}\label{Sec.ProbFormulation}
	\subsection{System Model}
	We consider a general nonlinear stochastic process whose state $X_k \in \mathbb{R}^{n_x}$ evolves according to the transition kernel
	\begin{align} \label{Eq.stateModel}
		X_k \sim p\!\left(X_k \middle| X_{k-1}\right).
	\end{align}
	The process state is observed by $m$ spatially distributed sensors, whose collective measurements at time $k$ are represented by the vector
	\begin{align} \label{Eq.obsModel}
		Y_k = [Y_k(i)]_{i=1}^m \sim p\left(Y_k \middle| X_k\right)
	\end{align}
	where $Y_k(i) \in \mathbb{R}^{n_y}$ denotes the local observation obtained from the $i$th sensor, which are correlated with the underlying process $X_k$. The measurement trajectory $\{Y_k\}_{k=1}^{K}$ is private and must be kept confidential. As illustrated in Fig.~\ref{Fig.SystemSetup}, a fusion policy aggregates the distributed sensory data and releases an instantaneous fused output $Z_k \in \mathbb{R}^{n_z}$ to a remote estimator, which aims at inferring $X_k$ for system monitoring and control. The output of the fusion center might also be used by an adversary to infer private information, which may lead to privacy leakage.
	\begin{figure}[h]
		\centering
		\includegraphics[width=0.5\textwidth]{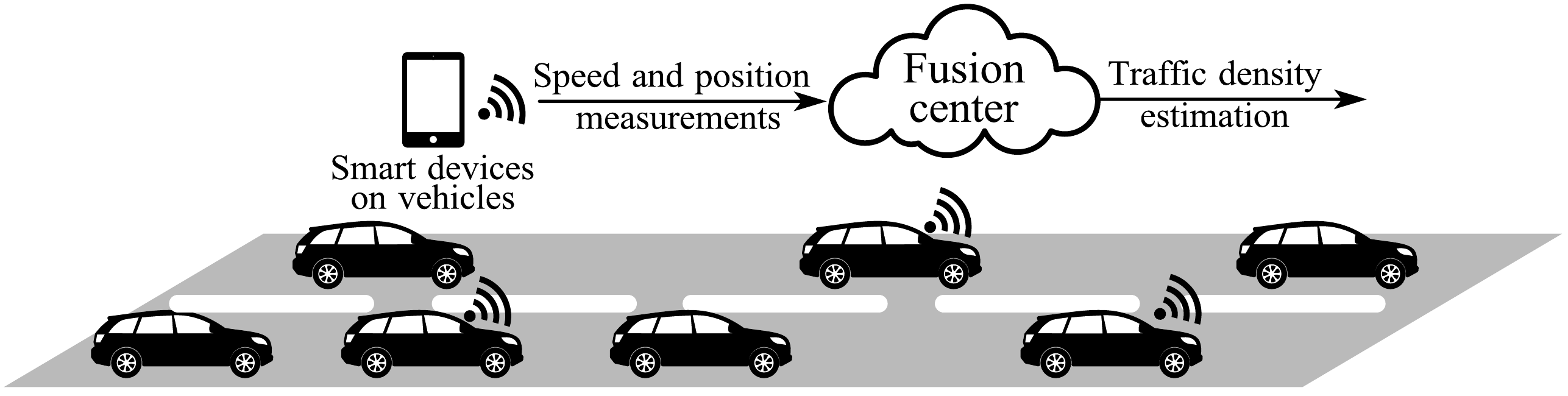}
		\caption{Traffic density estimation through position and speed measurements of vehicles}\label{Fig.trafficDensityEst}
	\end{figure}
	\subsection{Motivating Example}
	Traffic density on road networks is a key performance indicator for monitoring and control of transportation networks, which can be inferred using the position and speed measurements reported by smart devices, e.g., cell phones \cite{2019Fundamental}. As illustrated in Fig.~\ref{Fig.trafficDensityEst}, we consider a motivating scenario in which a subset of vehicles is equipped with similar smart devices that report their real-time GPS-based speeds and positions to a centralized fusion center for traffic density estimation. Fig.~\ref{Fig.realDensity} shows the estimated traffic density $\left(X_k\right)$ for a highway segment can be accurately estimated from the position and speed measurements $\left(Y_k\right)$. While this crowd-sourced sensing technique enables efficient congestion monitoring, the reported data may inadvertently reveal sensitive information about individual drivers, such as movement patterns or travel habits. To mitigate privacy risks, measurements from vehicles would be fused by a differentially private fusion technique prior to traffic density estimation.
		
	\begin{figure}[h]
		\centering
		\includegraphics[width=0.35\textwidth]{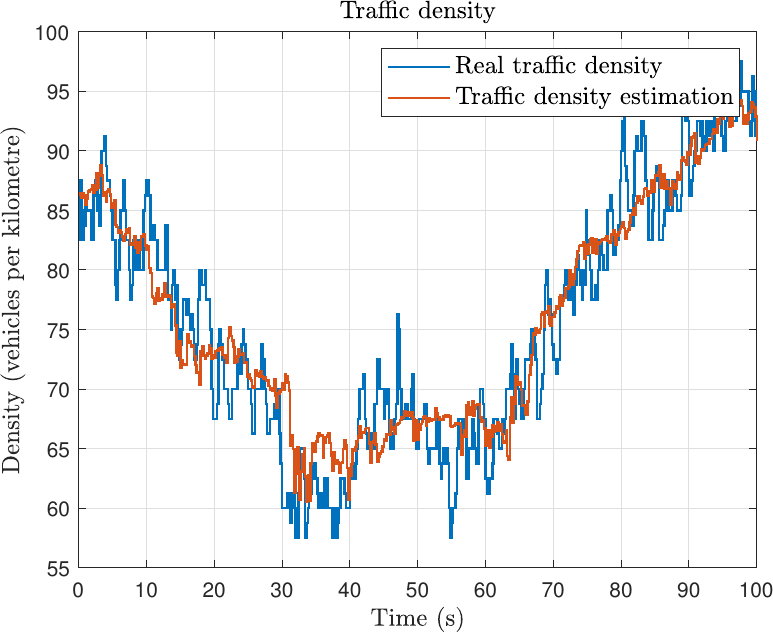}
		\caption{Traffic density in US-101 Highway dataset}\label{Fig.realDensity}
	\end{figure}
	\subsection{Privacy-Aware Fusion and Estimation}	
	At each time step $k$, given the historical measurements $Y^k=\left[Y_1, Y_2, \cdots Y^k\right]$ and past fusion results $Z^{k-1}=\left[Z_1, Z_2, \cdots Z^{k-1}\right]$, the stochastic fusion policy processes the sensor measurements according to
	\begin{align}
		Z_k \sim \mathcal{F}_k\!\left(Z_k \middle| Y^k, Z^{k-1}\right),
	\end{align}
	where $\mathcal{F}_k$ denotes the conditional probability distribution of the fusion output $Z_k$. The output of the fusion policy $Z_k$ is then utilized by a remote estimator to produce the state estimate $\tilde{X}_k$ according to
	\begin{align}
		\tilde{X}_k \sim \mathcal{E}_k\!\left(\tilde{X}_k \middle| Z^k, \tilde{X}^{k-1}\right),
	\end{align}
	where $\mathcal{E}_k$ represents the estimation policy at time $k$. For convenience, we define the fusion policy as $\mathcal{F} = \{\mathcal{F}_k\}_{k=1}^K$ and the estimation policy as $\mathcal{E} = \{\mathcal{E}_k\}_{k=1}^K$. 
	
	To prevent the adversary from inferring the private measurements from the released fusion results, the fusion policy is designed such that the total privacy leakage does not exceed the prescribed privacy budget $\mathsf{B_G} \in \mathbb{R}^{+}$.
	
	\subsection{Privacy Measure}
	We adopt the R\'enyi differential privacy (RDP) to quantify the privacy leakage of the sensitive measurements $Y^k$ through the fusion mechanism. At each time step $k$, given the fusion policy and past fusion results $Z^{k-1}$, the R\'enyi privacy loss is defined as
	\begin{align}\label{Eq.RenyiLossPerStep}
		L_k\!\left(\mathcal{F}_k, Z^{k-1}; \alpha \right)
		= \frac{1}{\alpha - 1}
		\log \sup_{Y^k, \tilde{Y}^k}
		\mathsf{E}_{\mathcal{F}_k\!\left( Z_k \middle| \tilde{Y}^k, Z^{k-1} \right)}
		\!\left[
		\left.
		\left(
		\frac{\mathcal{F}_k\!\left( Z_k \middle| Y^k, Z^{k-1} \right)}
		{\mathcal{F}_k\!\left( Z_k \middle| \tilde{Y}^k, Z^{k-1} \right)}
		\right)^{\alpha}
		\right| Z^{k-1}
		\right],
	\end{align}
	where $Y^k$ and $\tilde{Y}^k$ denote two adjacent measurement trajectories that differ in exactly one sensor’s measurement, \emph{i.e.}, $Y^k(j) = \tilde{Y}^k(j)$ for all $j \neq i$ and $Y^k(i) \neq \tilde{Y}^k(i)$ for some $1 \leq i \leq m$. According to \cite{mironov2017renyi}, the instantaneous fusion mechanism $\mathcal{F}_k$ satisfies $\big(\alpha, L_k\!\left(\mathcal{F}_k, Z^{k-1}; \alpha \right)\big)$-RDP. 
	
	The R\'enyi differential privacy characterizes the maximum divergence between the output distributions of a mechanism under neighboring datasets\cite{mironov2017renyi}. It serves as a natural relaxation of standard differential privacy and can be converted to the $(\epsilon, \delta)$-DP framework, while enabling more accurate and tractable privacy accounting over time. Notably, to enable privacy accounting via standard composition theorems \cite{dwork2014algorithmic, kairouz2015composition}, the privacy loss in most existing DP designs is fixed \emph{a priori}, and the mechanisms cannot adaptively adjust based on the released information. In contrast, the R\'enyi privacy loss defined in \eqref{Eq.RenyiLossPerStep} explicitly depends on the stochastic fusion outputs, which enables the fusion policy to adaptively regulate its output and allocate the privacy budget with greater flexibility.
	
	\subsection{Optimal Privacy-Aware Fusion Policy Design}
	Our objective is to develop an optimal fusion policy and the corresponding estimation policy that minimize the estimation error of system state while ensuring that the total privacy leakage does not exceed the prescribed privacy budget $\mathsf{B_G}$. To this end, we formulate the optimal privacy-aware fusion problem as the following constrained optimization problem,
	\begin{align}
		&\min_{\mathcal{F}} \; \min_{\mathcal{E}} \;
		\mathsf{E}\!\left[ \sum_{k=1}^{K} d\!\left( X_k, \tilde{X}_k \right) \right] \label{Eq.objEq} \\
		\mathrm{s.t.}\quad
		&\sum_{k=1}^{K} L_k\!\left(\mathcal{F}_k, Z^{k-1}; \alpha \right)
		\leq \mathsf{B_G}, \quad \forall Z_k=z_k \in \mathbb{R}^{n_z}, \; k = 1, \ldots, K, \label{Eq.constraint}
	\end{align}
	where $d\!\left(X_k, \tilde{X}_k\right)$ is a measurable function that quantifies the estimation error at time step $k$. The constraint  \eqref{Eq.constraint} ensures that the cumulative privacy leakage of each instantaneous funsion mechanism $\mathcal{F}_k$, given by $\sum_{k=1}^{K} L_k\!\left(\mathcal{F}_k, Z^{k-1}; \alpha \right)$, remains within the privacy budget $\mathsf{B_G}$. 
	
	According to the fully adaptive composition theorem in \cite{feldman2021individual}, the constraint \eqref{Eq.constraint} also yields trajectory-level RDP guarantees. Specifically, $\mathcal{F} = \{\mathcal{F}_k\}_{k=1}^K$ satisfies $(\alpha,\mathsf{B_G})$-RDP. Consequently, solving the proposed optimization problem yields an optimal fusion policy $\mathcal{F}$ that adaptively allocates the per-step budget $L_k$ to minimize the estimation error subject to the trajectory-level leakage budget $\mathsf{B_G}$.
	
	\section{Structural Properties of Optimal Privacy-Aware Fusion Design}\label{Sec.StructuralProperty}
	In this section, we study the structure properties of the proposed fusion design via the dynamic programming decomposition method. We derive the constrained optimality equations for the privacy-aware fusion optimization problem, and demonstrate that the optimal fusion policy adaptively allocates the privacy budget and regulates the adversary’s belief about the private measurements.
	\subsection{Constrained Optimality Equations}
	Let the the collection of instantaneous fusion policies be
	\begin{align} \label{Eq.policyCollection}
		\mathcal{C}_k\left(Z^{k-1}\right) = \left\{\mathcal{F}_k\left(z_k \middle| y^k, Z^{k-1}\right), \forall z_k \in \mathbb{R}^{n_z}, y_j \in \mathbb{R}^{m\times n_y}, j=1,2,\cdots k \right\} ,
	\end{align}
	which includes all conditional distributions that condition on the historical fusion results $Z^{k-1}$ and different measurements $Y^k=y^k$. Note that the R\'enyi privacy loss \eqref{Eq.RenyiLossPerStep} evaluates all instantaneous fusion polices that has the same conditional information $Z^{k-1}$, i.e., the collection \eqref{Eq.policyCollection}. Thus, the optimization problem \eqref{Eq.objEq}-\eqref{Eq.constraint} cannot be solved by obtaining the optimal individual sensor fusion policy $\mathcal{F}_k\left(\cdot \middle| Y^k=y^k, Z^{k-1}\right)$ when other polices $\mathcal{F}_k\left(\cdot \middle| Y^k\neq y^k, Z^{k-1}\right)$ in the collection $\mathcal{C}_k$ are unknown. In other words, we need to simultaneously solve all fusion polices with same $Z^{k-1}$, i.e., optimize the policy collection $\mathcal{C}_k$. 
	
	We next derive the constrained Bellman equations of the policy collection $\mathcal{C}=\left\{\mathcal{C}_k\right\}_{k=1}^K$ based on the optimality principle.
	\begin{theorem}\label{Th.OptEquations}
		The optimal differential private fusion collection $\mathcal{C}$ can be obtained via solving the following constrained optimality equations, 
		\begin{align}
			J_{k}^{\star}\left( s_k,b_k \right) =&\min_{\mathcal{C} _k} \mathsf{E}\left[ d\left( X_k,\tilde{X}_{k}^{\star} \right) +J_{k+1}^{\star}\left( s_{k+1},b_{k+1} \right) \middle| Z^{k-1} \right] \label{Eq.optEquationJ}
			\\
			\mathrm{s}.\mathrm{t}.,&\left\{ \begin{array}{c} \tilde{X}_{k}^{\star}=\mathrm{arg}_{\tilde{X}_k}\min \mathsf{E}\left[ d\left( X_k,\tilde{X}_k \right) \middle| Z^k \right], \\
				L_k\left( \mathcal{C} _k\left( b_k, s_k \right) ;\alpha  \right) \le s_k.\\
			\end{array} \right. \label{Eq.optConstraint}
		\end{align}
		where $s_k$ is the remaining privacy budget with $s_k=s_{k-1}-L_{k-1}\left( \mathcal{C} _{k-1}\left( b_{k-1}, s_{k-1} \right) ;\alpha  \right)$ and $s_1 = \mathsf{B_G}$, $b_k\left(X_k,Y_k\right)=p\left(X_k, Y^k \middle| Z^{k-1}\right)$ is the belief state with the update rule $b_{k+1} = \Phi\left(b_k, \mathcal{C}_k, Z_k\right)$, i.e.,
		\begin{align}
			b_{k+1}\left( X_{k+1},Y^{k+1} \right) =\frac{p\left( Y_{k+1} \middle| X_{k+1} \right) \int{p\left( X_{k+1} \middle| x_k \right) p\left( x_k,Y^k \middle| Z^{k-1} \right) dx_k}\mathcal{F} _k\left( Z_k \middle| Y^k,Z^{k-1},\mathsf{B}_k \right)}{\int{\int{p\left( x_k,y^k \middle| Z^{k-1} \right) \mathcal{F} _k\left( Z_k \middle| y^k,Z^{k-1},\mathsf{B}_k \right) dx_k}dy^k}}.
		\end{align}
		$J^\star$ is the optimal cost-to-go function with $J_{K+1} = 0$. The optimal fusion collection is $\mathcal{C}^\star _k\left(b_k, s_k \right)$.
	\end{theorem}	
	\begin{proof}
		See Appendix \ref{App.Th.OptEquations}.
	\end{proof}
	According to Theorem \ref{Th.OptEquations}, given the released fusion results $Z^{k-1}$, we can updated the belief state $b_k$ and the remaining privacy $s_k$, and compute the optimal fusion collection by solving the constrained optimality equations \eqref{Eq.optEquationJ}-\eqref{Eq.optConstraint}. Also, to meet the privacy budget constraint \eqref{Eq.constraint}, the per-step leakage $L_k\left(\mathcal{C}_k;\alpha\right)$ cannot exceed the remaining privacy budget \eqref{Eq.optConstraint}. 
	
	Base on the optimization result, given the measurements $Y^k$, the optimal fusion policy can be be further obtained via selecting the conditional distribution in $\mathcal{C}_k^\star$ that matches $Y^k$ as discussed as follows.
	\begin{lemma}\label{Lm.optFusionPolicy}
		Given the information $\left(Y^k=y^k, Z^{k-1}\right)$ and the optimal fusion collection $\mathcal{C}_k^\star$, the optimal fusion policy 
		$\mathcal{F}_k^\star\left(\cdot \middle| y^k, s_k, b_k\right)$
		is obtained by selecting the conditional distribution in $\mathcal{C}_k^\star$ that corresponds to the actual measurements $Y^k=y^k$.
	\end{lemma}
	\begin{proof}
		As discussed previously, the optimal fusion policy, denoted by $\mathcal{F}\left(\cdot \middle| Y^k = y^k, Z^k\right)$, belongs to the optimal fusion collection $\mathcal{C}_k^\star\left(Z^{k-1}\right)$. By Theorem \ref{Th.OptEquations}, $\mathcal{C}_k^\star$ depends only on $s_k$ and $b_k$, i.e., $\mathcal{C}_k^\star\left(s_k,b_k\right)$, where $s_k$ and $b_k$ are computed based on $Z^{k-1}$. Therefore, the optimal fusion policy can also be expressed as a function of $s_k$ and $b_k$, namely $\mathcal{F}\left(\cdot \middle| Y^k = y^k, s_k, b_k\right)$, and is selected from $\mathcal{C}_k^\star\left(s_k,b_k\right)$ according to the actual measurements $Y^k = y^k$.
	\end{proof}
	As established in Lemma~\ref{Lm.optFusionPolicy}, the optimal fusion policy depends not only on the measurements $Y^k$, but also on the remaining privacy budget $s_k$ and the belief state $b_k$. This feature distinguishes it from classical differentially private fusion policies based on the standard composition theorem. The structural distinction between these two classes of policies is illustrated in Fig.~\ref{Fig.optimalFusion}. In particular, the optimal fusion policy allocates the privacy budget adaptively in a closed-loop manner based on $b_k$ and $Y_k$. Moreover, conditioned on $Z^{k-1}$, the adversary can form its belief about the private measurements through the inference distribution $p\left(Y^k \middle| Z^{k-1}\right)$, which is the marginal of the belief state $b_k\left(X_k, Y^k\right)=p\left(X_k, Y^k \middle| Z^{k-1}\right)$. According to the belief state update rule, the evolution of the adversary's belief state is governed by the optimal fusion policy.
	
	
	In contrast, classical fusion policies usually rely on a predetermined per-step privacy budget $\mathsf{B}=\mathsf{B_G}/K$ and generate fusion outputs solely based on the measurements. Consequently, such policies cannot adaptively assign more privacy budget to some important time instance to improve data utility. By incorporating both the remaining privacy budget and the belief state, the optimal fusion policy offers greater flexibility in achieving a favorable privacy–utility trade-off. This advantage is further validated through numerical experiments in Sec.~\ref{Sec.Simulation}.
	\begin{figure*}
		\centering
		\includegraphics[width=0.5\textwidth]{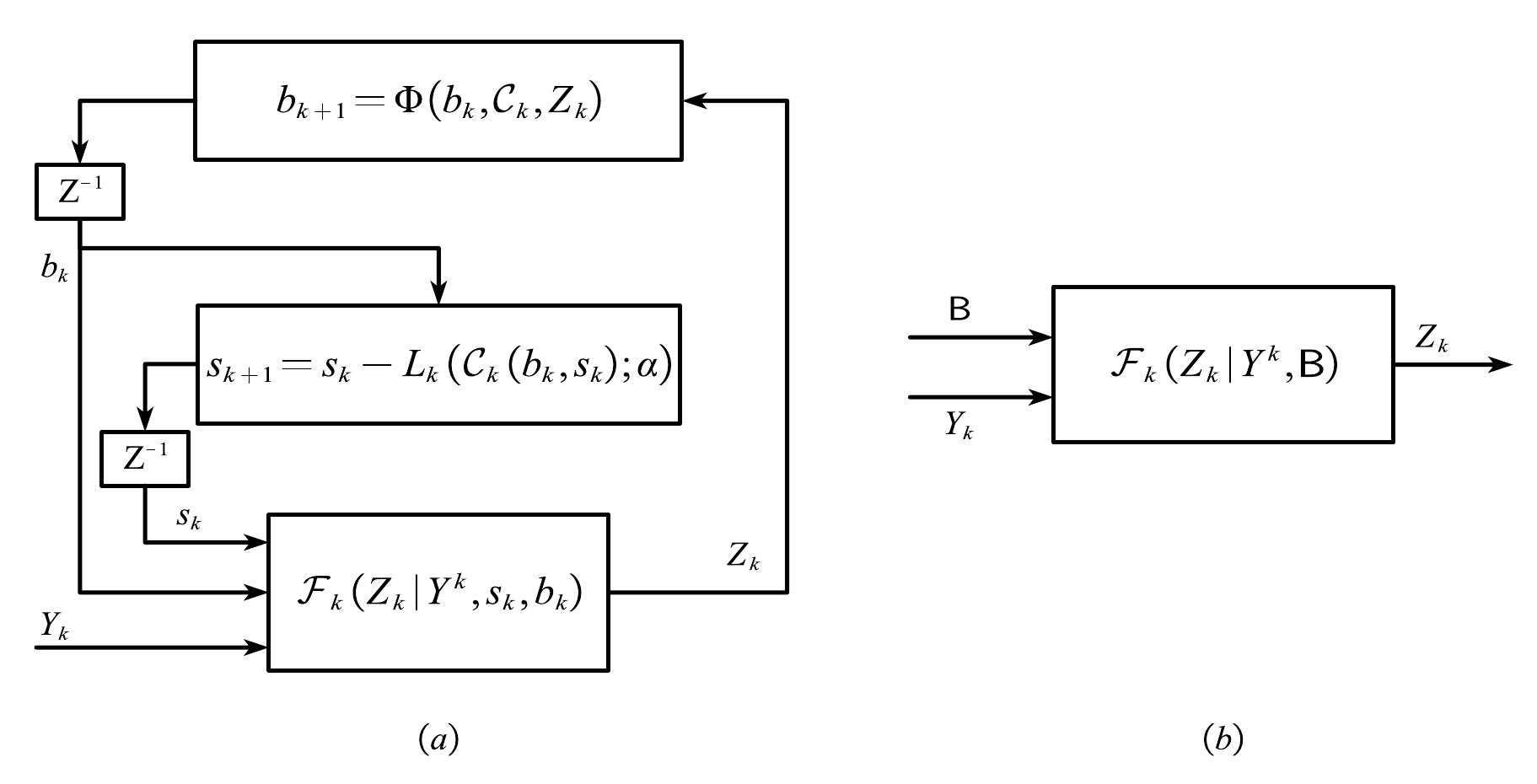}
		\caption{The structure of the optimal fusion policy.}\label{Fig.optimalFusion}
	\end{figure*}
	\section{Numerical Algorithm}\label{Sec.NumericalAlg}
	The optimal differentially private fusion design is computationally demanding, primarily due to the difficulty of evaluating the privacy leakage for general fusion policies \eqref{Eq.RenyiLossPerStep} and the curse of dimensionality inherent in the dynamic programming formulation of the constrained optimality equations \eqref{Eq.optEquationJ}--\eqref{Eq.optConstraint} \cite{powell2007approximate}. To address these challenges, we first parameterize the fusion policy using a structured conditional Gaussian distribution and derive a simple analytical expression for the R\'enyi privacy loss, which simplifies privacy evaluation. We then develop a numerical algorithm to optimize the parameterized fusion policy based on proximal policy optimization.
	
	\subsection{Policy Parameterization}
	We adopt a parametric recurrent representation to extract historical fusion information $Z^{k-1}$,
	\begin{align}
		h_{k,z} = \mathcal{H}_{\phi,z}\left(h_{k-1,z}, Z_{k-1}\right),
	\end{align}
	where $\mathcal{H}_{\phi,z}$ is a measurable function parameterized by $\phi$. Similarly, the measurements from the $i$-th sensor, $Y^k(i)$, are compressed via a recurrent representation,
	\begin{align}
		h_{k,y}(i) = \mathcal{H}_{\theta,y}\left(h_{k-1,y}(i), Y_k(i)\right),
	\end{align}
	and $h_{k,y} = [h_{k,y}(i)]_{i=1}^{m}$ extracts information of measurements from different sensors.
	
	The fusion policy is parameterized as a conditional Gaussian distribution:
	\begin{align}\label{Eq.paramFusion}
		\mathcal{F}_{\phi,\theta}\left( Z_k \middle| Y^k, Z^{k-1} \right)
		= \mathcal{N}\Big( Z_k;, g_{\phi}^\top(h_{k,z}, s_k)f_{\theta}(h_{k,y}), I \Big),
	\end{align}
	where $f_{\theta}(h_{k,y}) = [f_{\theta}(h_{k,y}(i),i)]_{i=1}^{m} \in [0,1]^{m \times d}$ denotes the filtering function applied to each sensor, with each row  $f_{\theta}(h_{k,y}(i),i)$ representing a $d$-dimensional feature vector extracted from a different sensor. Moreover, $g_{\phi}(h_{k,z}, s_k) \in \mathbb{R}^{m}$ denotes the fusion vector with the $i$th element $g_{\phi}(h_{k,z}, s_k, i)$ computed based on $h_{k,z}$ and $s_k$. For clarity, the structure of the parameterized policy is illustrated in Fig.~\ref{Fig.paramFusion}. Specifically, measurements from different sensors are first processed by the filtering function $f_{\theta}$, then aggregated through the fusion vector $g_{\phi}$, and finally perturbed by Gaussian noise. The proposed fusion policy in \eqref{Eq.paramFusion} is inspired by the differentially private filter structures in \cite{le2013differentially, degue2022differentially}, which have been shown to achieve a better privacy–utility trade-off compared with schemes that inject noise prior to fusion.
	
	To sample the fusion result from \eqref{Eq.paramFusion}, we first draw Gaussian noise $N_k \sim \mathcal{N}(0, I)$ from a standard Gaussian distribution, and then generate $Z_k$ via
	\begin{align}\label{Eq.paramFusionPolicy}
		Z_k = g_\phi^\top(h_{k,z}, s_k) f_{\theta}(h_{k,y}) + N_k.
	\end{align}
	In contrast, most classical differentially private mechanisms are independent of the history of the released information and use fixed privacy parameters, which can be similarly parameterized as
	\begin{align}\label{Eq.paramFusionPolicyClassical}
		Z_k = \tilde{g}^\top f_{\theta}(h_{k,y}) + N_k,
	\end{align}
	where the fusion vector $\tilde{g}$ is constant which is computed a priori. Differently, our fusion policy adaptively adjusts its output based not only on the current measurements but also on the history of the released information, as illustrated in Fig.~\ref{Fig.paramFusion}. This structure aligns with the optimal privacy-aware fusion policy described in Theorem~\ref{Th.OptEquations}.
	
	Furthermore, the estimation policy is parameterized as $\mathcal{E}_\omega(\tilde{h}_{k,z})$, with the recurrent hidden state
	\begin{align}
		\tilde{h}_{k,z} = \tilde{\mathcal{H}}_{\omega,z}(\tilde{h}_{k-1,z}, Z_{k-1}).
	\end{align}
	This parameterization enables end-to-end optimization of the fusion and estimation policies while maintaining computational tractability.
	\begin{figure}
		\centering
		\includegraphics[width=0.6\textwidth]{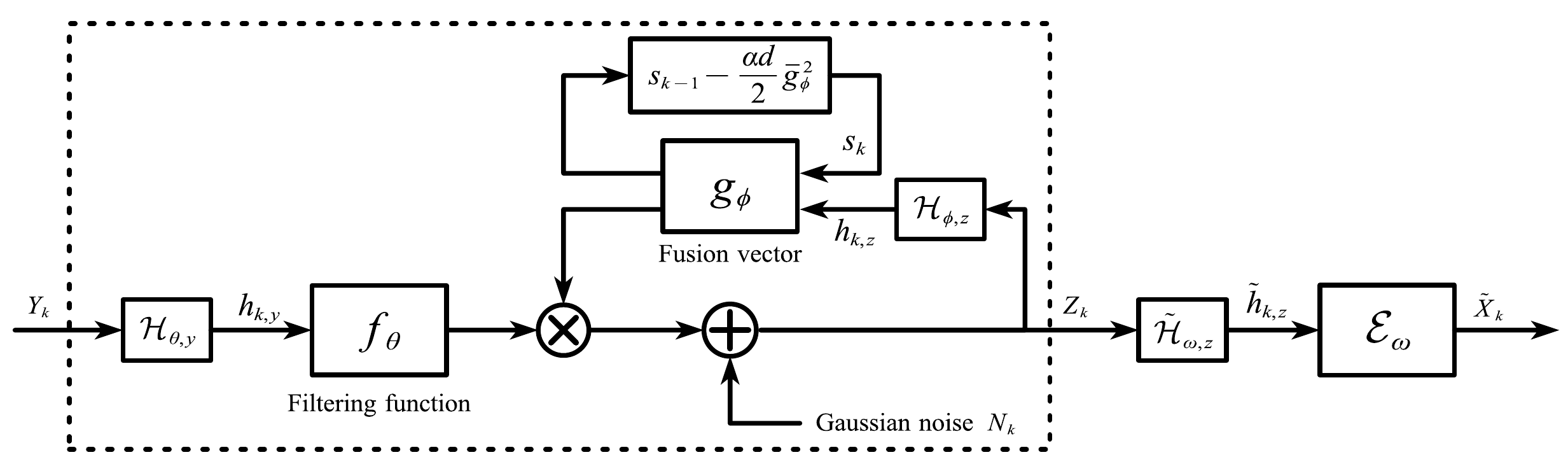}
		\caption{The structure of the parameterized fusion policy.}\label{Fig.paramFusion}
	\end{figure}
	\subsection{Privacy Leakage Evaluation}
	Evaluating the privacy metric in \eqref{Eq.RenyiLossPerStep} is computationally expensive for general fusion policies, as it involves divergence calculations over all possible adjacent measurement trajectories. In the following, however, we show that the privacy leakage of the fusion policy admits an analytical expression if it is parameterized as in \eqref{Eq.paramFusion}, thereby significantly reducing the computational complexity of policy optimization.
	\begin{theorem}\label{Th.instanPolicyRDP}
		Let $\bar{g}_{\phi}(h_{k,z}, s_k) = \max_i \left| g_{\phi}(h_{k,z}, s_k, i) \right|$. Then, the instantaneous R\'enyi privacy loss can be computed as
		\begin{align}
			L_k\left(\mathcal{F}_k, Z^{k-1}; \alpha \right) = \frac{\alpha d}{2} \, \bar{g}_{\phi}^2(h_{k,z}, s_k).
		\end{align}
		Furthermore, given the remaining privacy budget $s_k \ge 0$, the per-step leakage satisfies $L_k(\mathcal{F}_k, Z^{k-1}; \alpha) \le s_k$ if
		\begin{align}\label{Eq.rangeOfFusionVector}
			-\sqrt{\frac{2 s_k}{\alpha d}} \le \bar{g}_{\phi}(h_{k,z}, s_k) \le \sqrt{\frac{2 s_k}{\alpha d}}.
		\end{align}
	\end{theorem}
	
	\begin{proof}
		See Appendix~\ref{App.Th.instanPolicyRDP}.
	\end{proof}
	
	According to Theorem~\ref{Th.instanPolicyRDP}, the privacy leakage through $\mathcal{F}_k$ remains within the remaining budget as long as the fusion vector does not amplify the filtered measurements excessively relative to the Gaussian noise $N_k$, i.e., $\bar{g}_\phi$ is appropriately bounded. Notably, the constraint in \eqref{Eq.rangeOfFusionVector} varies with the historical fusion results $Z^{k-1}$ encoded in $h_{k,z}$, providing the fusion policy with greater flexibility to balance estimation error \eqref{Eq.objEq} against privacy constraints \eqref{Eq.constraint}. Moreover, for any fusion policy, the privacy budget constraint \eqref{Eq.optConstraint} is guaranteed if the fusion vector is clipped within the range specified by \eqref{Eq.rangeOfFusionVector}, i.e.,
	\begin{align}\label{Eq.fusionVecClip}
		\text{clip}\big(g_{\phi}(h_{k,z}, s_k), -\sqrt{\frac{2 s_k}{\alpha d}}, \sqrt{\frac{2 s_k}{\alpha d}} \big)
		= \min \Big\{ \max \big\{ -\sqrt{\frac{2 s_k}{\alpha d}}, g_{\phi}(h_{k,z}, s_k) \big\}, \sqrt{\frac{2 s_k}{\alpha d}} \Big\}.
	\end{align}
	
	\subsection{Sensor Fusion and State Estimation Under the Proposed Scheme}
	Motivated by Theorem~\ref{Th.instanPolicyRDP}, we implement the sensor fusion and state estimation algorithm summarized in Algorithm~\ref{Alg.ExecutionPolicis}. In particular, the fusion center executes Steps 3–6 to produce the fusion output $Z_k$ and enforce the hard privacy budget constraint in \eqref{Eq.optConstraint} via the clipping operation in \eqref{Eq.fusionVecClip}. After receiving $Z_k$, the estimator executes Step 8 to estimate the state. Since the privacy budget constraint is strictly enforced at all times, Algorithm~\ref{Alg.ExecutionPolicis} satisfies $\left(\alpha, \mathsf{B_G}\right)$-RDP and thus guarantees trajectory-level differential privacy \cite{feldman2021individual}.
	
	Furthermore, Algorithm~\ref{Alg.ExecutionPolicis} can be used to evaluate the parameterized fusion policy $\mathcal{F}_{\phi,\theta}$ and the state estimator $\mathcal{E}_\omega$, which is necessary for applying gradient descent techniques to update their parameters. In the following, we describe in detail how to optimize the fusion policy and state estimator based on Algorithm~\ref{Alg.ExecutionPolicis}.
	\begin{algorithm}
		\caption{$\left(\alpha, \mathsf{B_G}\right)$-RDP sensor fusion and state estimation}\label{Alg.ExecutionPolicis}
		\begin{algorithmic}[1]
			\State Initialize $k=1$, $s_1=\mathsf{B_G}$, $Z_0=0$, $\mathsf{B}_0=0$, $h_{k,y}$, $h_{k,z}$ and $\tilde{h}_{k,z}$.
			\While{$k\leq K$}
				\State Collect the measurement $Y_k$ 
				\State Update the recurrent representations $h_{k,y}$ and $h_{k,z}$ 
				\State Compute the fusion vector $g_{\phi}\big( h_{k,z},s_k \big)$ and clip it using \eqref{Eq.fusionVecClip} 
				\State Release the fusion result: $Z_k = g_\phi^\top\left(h_{k,z}, \mathsf{B}_k\right)f_{{\theta}}\left(h_{k,y}\right)+N_k$, $N_k \sim \mathcal{N}\left(0,I\right)$
				\State Update the remaining privacy budget $s_{k+1} =s_k-\frac{\alpha d}{2}\bar{g}_{\phi}^{2}\left( h_{k,z},\theta \right)$
				\State Update the recurrent representation $\tilde{h}_{k,z}$ and compute the state estimate $\mathcal{E}_\omega\left(\tilde{h}_{k,z}\right)$
				\State The environment returns the estimation error $d\left(X_k, \tilde{X}_k\right)$
				\State $k=k+1$
			\EndWhile
		\end{algorithmic}
	\end{algorithm}
	\subsection{Joint-optimization for Filtering Function and Estimator}
	As illustrated in Fig.~\ref{Fig.paramFusion}, the measurement and fusion information propagate forward through the deterministic filtering function to the state estimator, while the fusion vector is influenced backward by past fusion results. Consequently, when the fusion vector is fixed, the filtering function $f_\theta$ and the state estimator $\mathcal{E}_\omega$ can be jointly optimized by minimizing the sample mean of the estimation error. Specifically, we first sample $M$ trajectories $\{X^{K,(i)}, \tilde{X}^{K,(i)}\}_{i=1}^M$ according to Algorithm~\ref{Alg.ExecutionPolicis}, and then update $\theta$ and $\omega$ via
	\begin{align}\label{Eq.jointObj}
		(\omega ,\theta )\gets \,(\omega ,\theta )-\alpha \frac{1}{M}\sum_{i=1}^M{\sum_{k=1}^K{d\bigl( X_{k}^{(i)},\tilde{X}_{k}^{(i)} \bigr)}},
	\end{align}
	where the equation \eqref{Eq.jointObj} approximates \eqref{Eq.objEq} using the sample mean over the trajectories, and $\alpha$ is the step size.
	
	\subsection{Proximal Policy Optimization for Fusion Vector Function}
	As indicated in Theorem~\ref{Th.OptEquations}, the optimal fusion vector function $g_\phi$ design is a dynamic programming problem.	Proximal Policy Optimization (PPO) is a robust and sample-efficient reinforcement learning approach for such dynamic programming problems. Here, we show how to optimize $g_\phi$ via PPO algorithm while $f_\theta$ and $\mathcal{E}_\omega$ are fixed via \eqref{Eq.jointObj}. We follow the standard PPO framework \cite{schulman2017proximal}, defining the fusion vector as an action which is sampled from the stochastic policy,
	\begin{align}\label{Eq.stochaG}
		a_k \sim \mathcal{N}\big(a_k; g_\phi(h_{k,z}, s_k), \Sigma_\phi\big),
	\end{align}
	where $\Sigma_\phi$ is a learnable variance controlling the exploration–exploitation trade-off. This stochastic policy is denoted $\mathcal{A}_\phi(a_k | s_k, h_{k,z})$.
	
	Given information $(Z^{k-1}, s_k)$, the cost-to-go function of the fusion vector policy is
	\begin{align}\label{Eq.ACVF}
		J(Z^{k-1}, s_k) = \mathsf{E}\big[ Q(a_k, Z^{k-1}, s_k) \,|\, Z^{k-1}, s_k \big],
	\end{align}
	where the $Q$-function evaluates the value of action $a_k$,
	\begin{align}\label{Eq.ACQF}
		Q(a_k, s_k, Z^{k-1}) = \mathsf{E}\big[ d(X_k, \tilde{X}_k) + J(s_{k+1}, Z^k) \,|\, a_k, s_k, Z^{k-1} \big],
	\end{align}
	with $J(s_{K+1}, Z^{K}) = 0$. The corresponding advantage function is
	\begin{align}
		G(a_k, s_k, Z^{k-1}) = Q(a_k, s_k, Z^{k-1}) - J(s_k, Z^{k-1}),
	\end{align}
	quantifying the relative benefit of selecting action $a_k$.
	
	Following the actor–critic framework, the cost-to-go function is parameterized with $\psi$ and estimated via minimizing the mean-square error:
	\begin{align}
		\psi \gets \,\psi -\beta \frac{1}{M}\sum_{i=1}^M{ \sum_{k=1}^K{\bigl( J_{\psi}(s_{k}^{\left( i \right)},\hat{h}_{k,z}^{\left( i \right)})-\sum_{j=k}^K{d(X_{j}^{\left( i \right)},\tilde{X}_{j}^{\left( i \right)})} \bigr) ^2} },
	\end{align}
	with recursive representation $\hat{h}_{k,z} = \hat{\mathcal{H}}_\psi(\hat{h}_{k-1,z}, Z_{k-1})$ and the positive step size $\beta$.
	
	In the PPO framework, we denote the current fusion vector policy to be updated as $\mathcal{A}_\phi$, and the previous policy before current update as $\mathcal{A}_{\phi, \mathrm{old}}$. The probability ratio between the current and previous policies is defined as
	\begin{align}
		r_k(\phi) = \frac{\mathcal{A}_\phi(a_k \mid s_k, h_{k,z})}{\mathcal{A}_{\phi, \mathrm{old}}(a_k \mid s_k, h_{k,z})}.
	\end{align}
	We then sample trajectories from Algorithm~\ref{Alg.ExecutionPolicis} with $\mathcal{A}_{\phi, \mathrm{old}}$ and compute the advantage function $G_{\mathrm{old}}(a_k, s_k, Z^{k-1})$, which evaluates the relative benefit of action $a_k$ under the previous policy. 
	
	The policy parameters $\phi$ are updated by optimizing the following surrogate objective function:
	\begin{align}
		\phi \gets \phi -\gamma \frac{1}{M}\sum_{i=1}^M{\sum_{k=1}^K{\min \bigl( r_k(\phi )\,G_{\mathrm{old}}(a_{k}^{\left( i \right)},s_{k}^{\left( i \right)},Z^{k-1,\left( i \right)}),\;H_{\mathrm{old}}(\epsilon ,G_{\mathrm{old}},a_{k}^{\left( i \right)},s_{k}^{\left( i \right)},Z^{k-1,\left( i \right)};\phi ) \bigr)}},
	\end{align}
	where the clipped term is defined as
	\begin{align}
		H_{\mathrm{old}}(\epsilon, G_{\mathrm{old}}, a_k, s_k, Z^{k-1}; \phi) = \mathrm{clip}\big( r_k(\phi), 1-\epsilon, 1+\epsilon \big) \, G_{\mathrm{old}}(a_k, s_k, Z^{k-1}).
	\end{align}
	In this formulation, $r_k(\phi) G_{\mathrm{old}}$ measures the improvement of the current policy relative to the previous one, while the clipping operation constrains the policy ratio within $[1-\epsilon, 1+\epsilon]$. This ensures stable and robust policy updates, preventing overly large steps that may degrade performance.
	
	Finally, leveraging the numerical techniques described above, Algorithm~\ref{Alg.optAlg} is proposed to alternatively optimize the filtering function and estimator jointly, and optimize the fusion vector using PPO.
	
	\begin{algorithm}
		\caption{Alternative Optimization of Fusion Policy and State Estimator}\label{Alg.optAlg}
		\begin{algorithmic}
			\State Initialize the fusion function $g_\phi$, filtering function $f_\theta$ and the state estimator $\mathcal{E}_\omega$. 
			\Repeat
			\Repeat
			\State Collect data and evaluate the filtering function $f_\theta$ and state estimator $\mathcal{E}_\omega$ via Algorithm \ref{Alg.ExecutionPolicis}.
			\State Update $f_\theta$ and $\mathcal{E}_\omega$ via \eqref{Eq.jointObj} with the fixed fusion function $g_\phi$.
			\Until convergence 
			\Repeat
			\State Collect data and evaluate $g_\phi$ via Algorithm \ref{Alg.ExecutionPolicis}.
			\State Update $g_\phi$ via PPO algorithm with fixed $f_\theta$ and $\mathcal{E}_\omega$.
			\Until convergence 
			\Until convergence
		\end{algorithmic}
	\end{algorithm}
	\begin{figure}[h]
		\centering
		\includegraphics[width=0.5\textwidth]{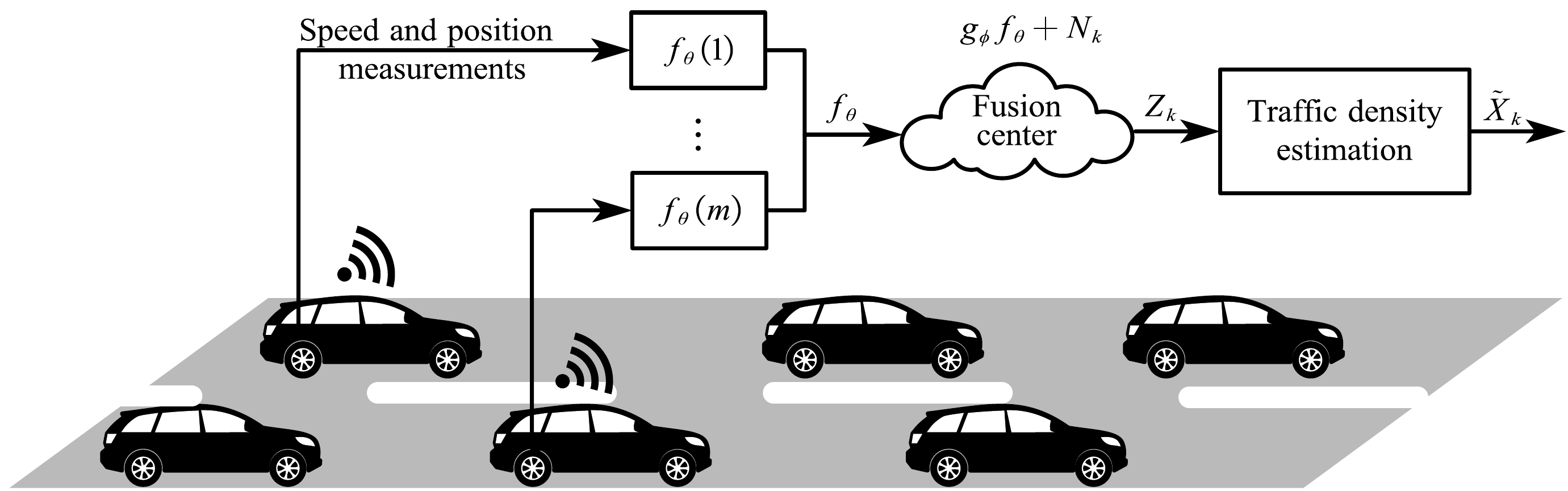}
		\caption{Traffic density estimation based on the proposed $\left(\alpha, \mathsf{B}_\mathsf{G}\right)$-RDP fusion scheme}\label{Fig.trafficDensityFusion}
	\end{figure}
	\section{Numerical Results}\label{Sec.Simulation}
	In this section, we apply the proposed privacy-aware fusion framework to estimate the traffic density $X_k$ on a road segment, where $X_k$ is defined as the average number of vehicles per kilometer. The estimation is based on the position and velocity data $Y_k$ reported by vehicles traveling on the road. Specially, as shown in Fig.~\ref{Fig.trafficDensityFusion}, the measurement from each vehicle is first processed by $f_\theta$ and then fused to produce differentially private output $Z_k$, which is subsequently used for traffic density estimation. Since the sensor fusion and state estimation procedure is carried out according to Algorithm~\ref{Alg.ExecutionPolicis}, the measurement data from different vehicles is protected with $\left(\alpha, \mathsf{B_G}\right)$-RDP guarantees at the trajectory level.
	
	\subsection{US Highway 101 Dataset}	
	To evaluate the robustness and effectiveness of Algorithm~\ref{Alg.optAlg}, we conduct simulations based on the US Highway 101 dataset \cite{simulation2007us}, which contains 45 minutes of detailed vehicle trajectory data on southbound US 101. We select a 400-meter segment as the study area, set the sampling interval to 0.2s, and consider a horizon $K=100$. 
	
	We assume that a subset of vehicles on the road report their speed and position measurements to the fusion center for traffic prediction and management. Meanwhile, the ground-truth traffic density is computed by counting the number of vehicles within the study area at each time step and normalizing by the segment length for estimation comparison. The estimation performance of the proposed fusion scheme is evaluated using the squared error, defined as $d\left(X_k, \tilde{X}_k\right) = \left( X_k - \tilde{X}_k \right)^2$.
	
	The dataset is divided into training, validation, and testing subsets. Specifically, 80\% of the data is used for training the fusion policy, while the remaining data is split equally, with 10\% used for validation and 10\% for testing.
	\begin{figure}[h]
		\centering
		\includegraphics[width=0.35\textwidth]{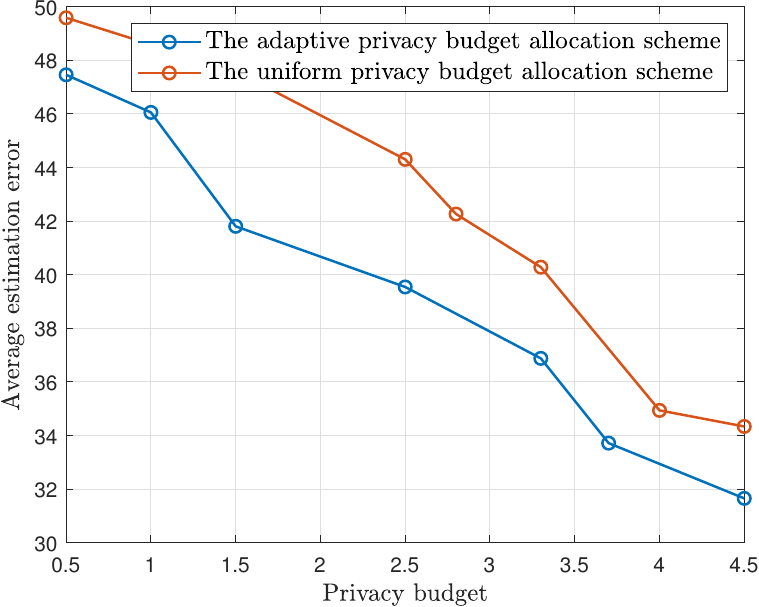}
		\caption{Average estimation error versus the privacy budget $\mathsf{B_G}$ under the adaptive fusion design and the classical fusion design.}\label{Fig.renyiTradeoff}
	\end{figure}
	\subsection{Privacy-Aware Traffic Density Estimation}
	We compare the proposed adaptive privacy-aware fusion design \eqref{Eq.paramFusionPolicy} with the classical differentially private mechanism \eqref{Eq.paramFusionPolicyClassical} that uniformly distributes the privacy budget over time. The proposed design is optimized using Algorithm~\ref{Alg.optAlg}, while the classical mechanism is optimized by minimizing the mean squared error \eqref{Eq.jointObj}.
	
	To assess the privacy–utility trade-off, we plot the average estimation error versus the total privacy budget $\mathsf{B_G}$ in Fig.~\ref{Fig.renyiTradeoff}. As shown in Fig.~\ref{Fig.renyiTradeoff}, the estimation error increases as the privacy budget decreases, reflecting the trade-off between higher privacy level and lower estimation accuracy. Specifically, a smaller privacy budget limits the maximum fusion vector magnitude $\bar{g}_\phi$ according to \eqref{Eq.rangeOfFusionVector}, reducing the signal-to-noise ratio of the fused output $Z_k$ and increasing the estimation error.
	\begin{figure}
		\centering
		\subfigure[]{
			\centering
			\includegraphics[width=0.35\textwidth]{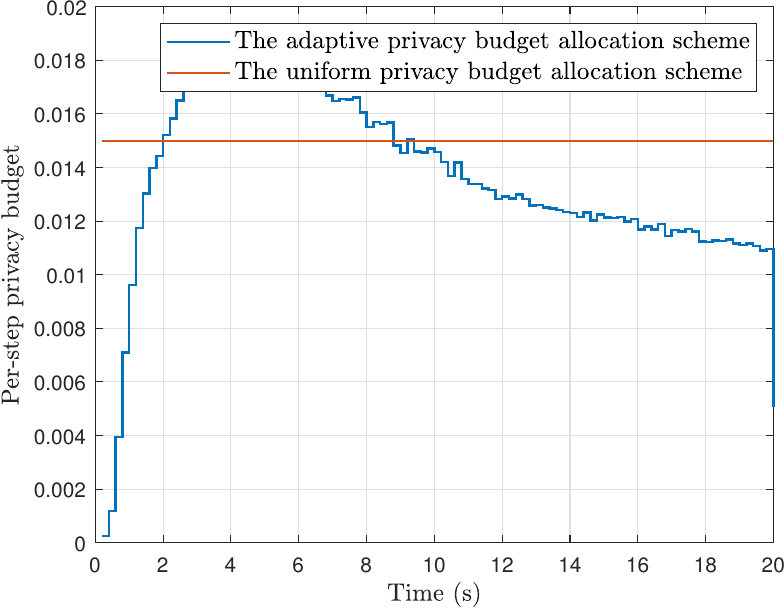}\label{Fig.budgetAllocation}	
		}
		\subfigure[]{
			\centering
			\includegraphics[width=0.35\textwidth]{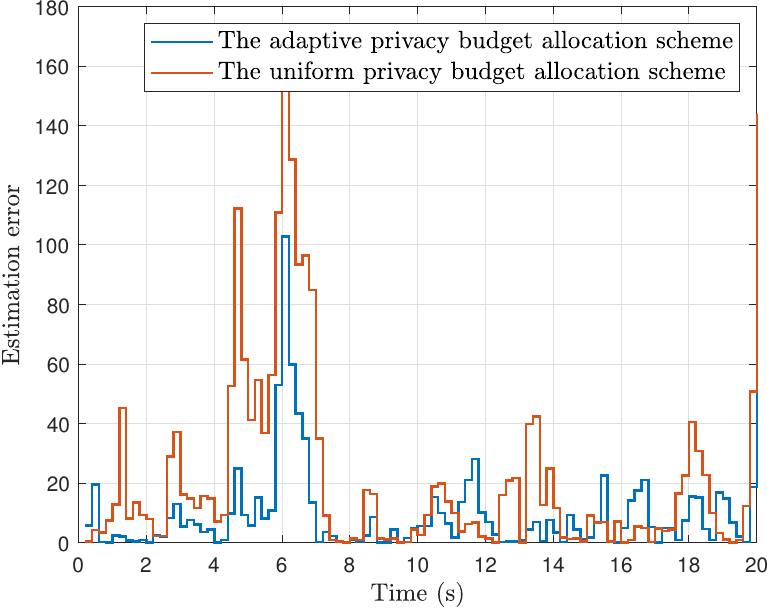}\label{Fig.estimationError}	
		}
		\caption{(a) Privacy budget allocation, (b) estimation error under the adaptive fusion design and the classical fusion design with the budget $\mathsf{B_G}=1.5$.}
	\end{figure}
	\begin{figure}
		\centering
		\includegraphics[width=0.35\textwidth]{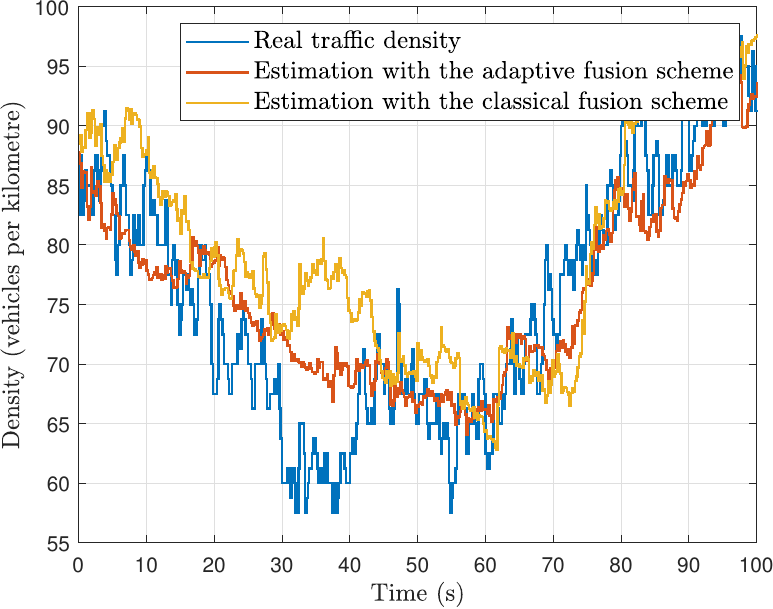}
		\caption{Traffic density estimation with privacy budget $\mathsf{B_G}=1.5$ at each time horizon.}\label{Fig.estDensityCmp}
	\end{figure}
	
	Specifically, the proposed adaptive fusion design achieves lower estimation error than the classical mechanism under the same privacy budget by dynamically allocating the budget over time. We compare the per-step allocated budgets and the corresponding estimation errors for the two fusion designs with a total budget of $\mathsf{B_G} = 1.5$ in Fig.~\ref{Fig.budgetAllocation} and Fig.~\ref{Fig.estimationError}, respectively. The adaptive design allocates more budget in the 2--8s interval, leading to lower estimation error compared to the classical design. 
	
	Moreover, we draw the estimated traffic density trajectories under two fusion designs in Fig.~\ref{Fig.estDensityCmp}. As observed, the trajectory obtained with adaptive budget allocation more closely matches the true traffic density, demonstrating improved estimation accuracy over the uniform budget allocation.
	\section{Conclusion}\label{Sec.Conclusion}
	In this paper, we proposed an optimal privacy-aware sensor fusion framework for general nonlinear systems using R\'enyi differential privacy. We analyzed the structural properties of the optimal fusion design through the derived constrained optimality equations and showed that the fusion policy adaptively allocates the privacy budget while minimizing estimation error. Furthermore, we developed a practical numerical algorithm to optimize the parameterized fusion policy and jointly train the state estimator. The effectiveness of the proposed approach was demonstrated through traffic density estimation on a real-world US Highway 101 dataset, where the adaptive privacy-aware fusion achieved improved estimation accuracy compared to classical differentially private mechanisms. 
	\appendices
	\section{Proof of Theorem \ref{Th.OptEquations}}\label{App.Th.OptEquations}
	To prove Theorem \ref{Th.OptEquations}, we first define the remaining privacy budget and belief state, and show that they can be updated in a Bayesian manner.
	
	\begin{lemma}\label{Lm.stateUpdate}
		Let $s_k$ denote the remaining privacy budget at time $k$ with $s_1=\mathsf{B_G}$, and let $b_k$ denote the conditional distribution
		\[
		b_k = p\left(X_k, Y^k \middle| Z^{k-1}\right)
		\]
		with $b_1(X_1,Y_1)=p(X_1,Y_1)$. Then, $s_{k+1}$ and $b_{k+1}$ can be updated as
		\begin{align}
			&s_{k+1} = s_k - L_k\left(\mathcal{C}_k(Z^{k-1});\alpha\right),\\
			b_{k+1}(X_{k+1}, Y^{k+1}) &= \frac{p(Y_{k+1} \mid X_{k+1}) \int p(X_{k+1} \mid x_k) p(x_k, Y^k \mid Z^{k-1}) dx_k \, \mathcal{F}_k(Z_k \mid Y^k, Z^{k-1})}{\int \int p(x_k, y^k \mid Z^{k-1}) \mathcal{F}_k(Z_k \mid y^k, Z^{k-1}) dx_k dy^k},
		\end{align}
		where the update of belief state $b_k$ is denoted as $b_{k+1} = \Phi(b_k, Z_k, \mathcal{C}_k)$.
	\end{lemma}
	
	\begin{proof}
		By definition, the remaining budget is initialized as $s_1 = \mathsf{B_G}$, and at each step the consumed privacy is $L_k(\mathcal{C}_k(Z^{k-1});\alpha)$, giving
		\[
		s_{k+1} = s_k - L_k(\mathcal{C}_k(Z^{k-1});\alpha).
		\]
		
		For the belief state, we have
		\[
		b_{k+1}(X_{k+1}, Y^{k+1}) = p(X_{k+1}, Y^{k+1} \mid Z^k) = \frac{p(X_{k+1}, Y^{k+1}, Z_k \mid Z^{k-1})}{p(Z_k \mid Z^{k-1})}.
		\]
		The numerator can be expanded as
		\[
		p(X_{k+1}, Y^{k+1}, Z_k \mid Z^{k-1}) = \int p(Y_{k+1} \mid X_{k+1}) p(X_{k+1} \mid x_k) p(x_k, Y^k \mid Z^{k-1}) \mathcal{F}_k(Z_k \mid Y^k, Z^{k-1}) dx_k,
		\]
		and the denominator $p(Z_k \mid Z^{k-1})$ is the marginal over $X_{k+1}, Y^{k+1}$. This gives the update formula in the lemma. The update depends on the fusion result $Z_k$ and the policy collection $\mathcal{C}_k(Z^{k-1})$, which is summarize as $b_{k+1} = \Phi(b_k, Z_k, \mathcal{C}_k)$.
	\end{proof}
	
	\noindent
	We now prove Theorem \ref{Th.OptEquations}. Given a fusion policy, the optimal estimator is
	\[
	\tilde{X}_k^\star = \arg\min_{\tilde{X}_k} \mathsf{E}\big[d(X_k, \tilde{X}_k) \mid Z^k \big],
	\]
	since the estimator does not affect future fusion decisions or system dynamics.
	
	The privacy constraint
	\[
	\sum_{k=1}^{K} L_k(\mathcal{C}(Z^{k-1});\alpha) \le \mathsf{B_G}
	\]
	defines the admissible policy collection
	\begin{align} \label{Eq.admisPolicy}
		\mathcal{C}_k(Z^{k-1}) \in \Big\{ \mathcal{C}_k(Z^{k-1}) : L_k(\mathcal{C}(Z^{k-1}); \alpha) \le s_k \Big\}, \quad \forall k=1,\dots,K,
	\end{align}
	with $s_{k+1} = s_k - L_k\left(\mathcal{C}_k(Z^{k-1});\alpha\right)$.
	
	Base on the admissible policy collection, we define the optimal cost-to-go function at step $k$ as
	\[
	J^\star_k(Z^{k-1}) = \min_{\text{admissible }\{\mathcal{C}_j\}_{j=k}^K} \mathsf{E} \Big[ \sum_{j=k}^{K} d(X_j, \tilde{X}_j^\star) \mid Z^{k-1} \Big].
	\]
	At time step $K$, we have
	\begin{align}\label{Eq.optCostK}
		J^\star_K(Z^{K-1}) = \min_{\text{admissible }\mathcal{C}_K} \mathsf{E} \Big[ d(X_K, \tilde{X}_K^\star) \mid Z^{K-1} \Big],
	\end{align}
	where the expectation $p\left(X_K, Z_K \middle| Z^{K-1}\right)$ can be written as 
	\begin{align}\label{Eq.expProb}
		p\left( X_K,Z_K\middle| Z^{K-1} \right) &
		=\int{p\left( X_K,y^K,Z_K\middle| Z^{K-1} \right) dy^K} \nonumber
		\\
		&=\int{\mathcal{F} _K\left( Z_K\middle| y^K,Z^{K-1} \right) p\left( X_K,y^K\mid Z^{K-1} \right) dy^K} \nonumber
		\\
		&=\int{\mathcal{F} _K\left( Z_K\middle| y^K,Z^{K-1} \right) b_K\left( X_K,y^K \right) dy^K},
	\end{align}
	Hence, this distribution is determined by the collection $\mathcal{C}_K =\left\{\mathcal{F}_K\left(z_K \middle| y^K, Z^{K-1}\right), \forall z_K \in \mathbb{R}^{n_z}, y_j \in \mathbb{R}^{m\times n_y}, j=1,2,\cdots K \right\}  $ and the belief state $b_K$. In addition, the admissible policy collection $\mathcal{C}_K$ is constrained by the remaining privacy budget $s_K$, as defined in \eqref{Eq.admisPolicy}. Therefore, the optimization problem in \eqref{Eq.optCostK} depends on both $s_K$ and $b_K$. Consequently, both $J^\star_K(Z^{K-1})$ and $\mathcal{C}_K^\star$ can be regarded as functions of $s_K$ and $b_K$.

	By induction, assume for $k+1 \le K$ we have
	\[
	J^\star_{k+1}(Z^k) = J^\star_{k+1}(s_{k+1}, b_{k+1}),
	\]
	with $\sum_{j=k+1}^{K} L_j(\mathcal{C}(Z^{j-1}); \alpha) \le s_{k+1}$. Using the principle of optimality, we obtain
	\begin{align}
		J^\star_k(Z^{k-1}) &= \min_{\text{admissible } \mathcal{C}_k} \mathsf{E}\big[ d(X_k, \tilde{X}_k^\star) + J^\star_{k+1}(Z^k) \mid Z^{k-1} \big] \nonumber \\
		&= \min_{\text{admissible } \mathcal{C}_k} \mathsf{E}\big[ d(X_k, \tilde{X}_k^\star) + J^\star_{k+1}(s_{k+1}, b_{k+1}) \mid Z^{k-1} \big] \nonumber \\
		&= \min_{\text{admissible } \mathcal{C}_k} \mathsf{E}\big[ d(X_k, \tilde{X}_k^\star) + J^\star_{k+1}(s_k - L(\mathcal{C}_k(Z^{k-1});\alpha), \Phi(b_k, Z_k, \mathcal{C}_k)) \mid Z^{k-1} \big],
	\end{align}
	where the last equality follows from Lemma \ref{Lm.stateUpdate}. The expectation is taken over $p(X_k, Z_k \mid Z^{k-1})$, which depends on $b_k$ similar to \eqref{Eq.expProb}. Hence, the optimal cost-to-go depends on $s_k$ and $b_k$, and the policy collection $\mathcal{C}_k^\star$ can be expressed as a function of $(s_k, b_k)$, yielding the constrained optimality equations in Theorem \ref{Th.OptEquations}.
	
	\section{Proof of Theorem \ref{Th.instanPolicyRDP}}\label{App.Th.instanPolicyRDP}
	Consider two adjacent measurement trajectories, $Y^k$ and $\tilde{Y}^k$, that differ only in the measurements of sensor $i$, i.e., $Y^k(i) \neq \tilde{Y}^k(i)$. Since each row of the filtering function processes only the measurements from a single sensor, it follows that
	\[
	f_\theta(h_{k,y}(i),i) \neq f_\theta(h_{k,\tilde{y}}(i),i) \quad \text{and} \quad f_\theta(h_{k,y}(j),j) = f_\theta(h_{k,\tilde{y}}(j),j), \; j \neq i.
	\]
	Moreover, with $f_{\theta}(h_{k,y}(i), i) \in [0,1]^d$, the largest possible deviation between the corresponding filtering outputs for two adjacent measurement trajectories is $$\abs{f_\theta(h_{k,y},i) - f_\theta(h_{k,\tilde{y}},i)}=\mathsf{1}_d,$$ where $\mathsf{1}_d$ is a $d$-dimensional vector of ones. 
	
	Thus, the $l_2$ sensitivity of the fusion policy \eqref{Eq.paramFusionPolicy} is
	\begin{align}
		\Delta_2(\mathcal{F}_k, Z^{k-1}) 
		&:= \sup_{Y^k, \tilde{Y}^k} \big\| g_\phi^\top(h_{k,z}, s_k) \big[ f_\theta(h_{k,y}) - f_\theta(h_{k,\tilde{y}}) \big] \big\|_2 \nonumber \\
		&= \max_i \| g_\phi(h_{k,z}, s_k, i) \mathsf{1}_d \|_2 \nonumber \\
		&= \max_i \sqrt{d} \, | g_\phi(h_{k,z}, s_k, i) | \nonumber \\
		&:= \sqrt{d} \, \bar{g}_\phi(h_{k,z}, s_k),
	\end{align}
	where $g_\phi(h_{k,z}, s_k, i)$ is the $i$th element of the fusion vector $g_\phi$. 
	
	By \cite{mironov2017renyi}, the R\'enyi privacy loss of the Gaussian mechanism is
	\[
	L_k(\mathcal{F}_k(\cdot \mid Z^{k-1}), \alpha) = \frac{\alpha}{2} \Delta_2^2(\mathcal{F}_k, Z^{k-1}) = \frac{\alpha d}{2} \bar{g}_\phi^2(h_{k,z}, s_k).
	\]
	Consequently, if
	\[
	-\sqrt{\frac{2 s_k}{\alpha d}} \le \bar{g}_\phi(h_{k,z}, s_k) \le \sqrt{\frac{2 s_k}{\alpha d}},
	\]
	then $L_k(\mathcal{F}_k,Z^{k-1}; \alpha) \le s_k$, satisfying the per-step RDP constraint.
	
	\bibliographystyle{ieeetr}
	\bibliography{reference}
\end{document}